\documentclass[letter,12pt]{article}

\usepackage[USenglish]{babel} 
\usepackage[T1]{fontenc}      
\usepackage[utf8]{inputenc}   
 \usepackage{upquote}          
 \usepackage{csquotes}         
 \usepackage{textcmds}         
 \usepackage{paralist}         
\usepackage{graphicx}

\usepackage{amsmath}
\usepackage{amsthm}
\usepackage{amssymb}
\usepackage{authblk}
\usepackage{bm}
\usepackage{mathtools}
\usepackage{braket}
\usepackage{framed}

\usepackage{tikz}
\usetikzlibrary{decorations.pathreplacing}

\usepackage{hyperref}

\usepackage{cleveref}

\usepackage{url}
\usepackage{xcolor}

\usepackage[textsize=scriptsize]{todonotes}

\usepackage{xspace}

\makeatletter
\newcommand*{\etc}{%
    \@ifnextchar{.}%
        {etc}%
        {etc.\@\xspace}%
}
\makeatother

\makeatletter
\newcommand*{\etal}{%
    \@ifnextchar{.}%
        {et al}%
        {et al.\@\xspace}%
}
\makeatother

\numberwithin{equation}{section}

\newtheorem{theorem}{Theorem}[section]

\newtheorem{lemma}[theorem]{Lemma}

\crefname{equation}{equation}{equations}

\newtheorem{corollary}[theorem]{Corollary}

\newtheorem{observation}[theorem]{Observation}

\theoremstyle{remark}
\newtheorem*{notation}{Notation}
\newtheorem*{observation*}{Observation}

\usepackage{stmaryrd} 
\usepackage{interval}
\intervalconfig{soft open fences}

\newlength{\arrow}
\newcommand{\Z}{\mathbb{Z}}

\newcommand{\bfN}{\textbf{N}_1}
\newcommand{\bfNt}{\textbf{N}_2}
\newcommand{\bfmod}{\textbf{ mod }}

\newcommand{\zs}{z_0,\dots,z_n}

\newtheorem{case}{Case}

\renewcommand{\qed}{$\square$}

\usepackage{graphicx}

\title{Population-Induced Phase Transitions and the Verification of Chemical Reaction Networks \thanks{This research was supported in part by National Science Foundation grants 1545028 and 1900716.}}

\author[]{James I. Lathrop}
\author[]{Jack H. Lutz}
\author[]{Robyn R. Lutz}
\author[]{Hugh D. Potter}
\author[]{Matthew R. Riley}
\affil[]{Iowa State University, USA}
\date{}

\begin{document}

\maketitle
\begin{abstract}
We show that very simple molecular systems, modeled as chemical reaction networks, can have behaviors that exhibit dramatic phase transitions at certain population thresholds.  Moreover, the magnitudes of these thresholds can thwart attempts to use simulation, model checking, or approximation by differential equations to formally verify the behaviors of such systems at realistic populations.  We show how formal theorem provers can successfully verify some such systems at populations  where other verification methods fail. 

\end{abstract}
\section{Introduction}

Chemical reaction networks, mathematical abstractions similar to Petri nets, are used as a programming language to specify the dynamic behaviors of engineered molecular systems.  Existing software can compile chemical reaction networks into DNA strand displacement systems that simulate them with growing generality and precision \cite{cSoSeWi09, jCDSPCS13, cBSJDTW17, jTTBJAC17}.  Programming is a challenging discipline in any case, but this is especially true of molecular programming, because chemical reaction networks---in addition to being Turing universal \cite{jSCWB08,oCSWB09, cFLBP17} and hence subject to all the uncomputable aspects of sequential, imperative programs--are, like the systems that they specify, distributed, asynchronous, and probabilistic.  Since many envisioned applications of molecular programming will be safety critical \cite{oWooLin05,jZhaSee11, jDoBaCh12, jLiLiYa13,jLJLZTS18, jSAGM20, oDiet20}, programmers thus seek to create chemical reaction networks that can be {\it verified} to correctly carry out their design intent.

One principle that is sometimes used in chemical reaction network design is the {\it small population heuristic} \cite{jLPCKP12, cCaKwWh16, jEKLLLM19}.  The idea here is to verify various stages of a design by model checking or software simulation to ferret out bugs in the design prior to laboratory experimentation or deployment.  Since the number of states of a molecular system is typically much larger than its population (the number of molecules present), and since molecular systems typically have very large populations, this model checking or simulation can usually only be carried out on populations that are far smaller than those of the intended molecular systems.  It is nevertheless reasonable to hope that, if a system is going to consist of a very large number of ``devices'' of various sorts, then any unforeseen errors in these devices' interactions will manifest themselves even with very small populations of each device.  It is this reasonable hope that is the underlying premise of the small population heuristic.  (Note that the small population heuristic can be regarded as a molecular version of the small scope hypothesis \cite{jJack19}.)

The question that we address here is whether real molecular systems can thwart the small population heuristic.  That is, can a real molecular system behave very differently at large populations than at small populations?  If so, {\it how sensitive} can its behavior be to its population, and {\it how simple a mechanism} can achieve such sensitivity?

In order to ensure that we are only investigating population effects, we focus our attention on chemical reaction networks that are {\it population protocols} in the sense that their populations remain constant throughout their operations.  If we have such a chemical reaction network, and if we vary its initial population {\it and nothing else}, then we are assured that any resulting variations of behavior are due solely to the differing populations.

In this paper we show that very simple chemical reaction networks can be very sensitive to their own populations.  In fact, they can exhibit {\it population-induced phase transitions}, behaving one way below a threshold population and behaving very differently above that threshold.  After reviewing chemical reaction networks in Section~\ref{sec:CRN}, we present in Section~\ref{sec:SinglePhaseTransition} a chemical reaction network $\mathbf{N}_1$, and we prove that $\mathbf{N}_1$ exhibits a population-induced phase transition in the following sense.  There are two parameters, $m$ and $n$, in the construction.  For this discussion, we may take $m = 34$ and $n = 67$, but the construction is general.  There are $n+1$ reactions among $n+2$ species (molecule types) in $\mathbf{N}_1$.  A species $Z_0$ is given an initial population of $p$, and all other species counts are initially 0.  Each reaction of $\mathbf{N}_1$ has two reactants and two products, so the total population of $\mathbf{N}_1$ is $p$ at all times.  There are in $\mathbf{N}_1$ two distinguished species, $B$ and $R$.  These ``blue'' and ``red'' species are abstract stand-ins for two different behaviors of $\mathbf{N}_1$.  Our construction exploits the inherent nonlinearity of chemical kinetics to ensure that, if $p < 2^m$, then $\mathbf{N}_1$ terminates with essentially all its population blue, while if $p\geq 2^m$, then $\mathbf{N}_1$ terminates with essentially all its population red.  Thus $\mathbf{N}_1$ exhibits a sharp phase transition at the population threshold $p = 2^m$.

Our construction is very simple.  The chemical reaction network $\mathbf{N}_1$ changes its behavior at the threshold $p = 2^m$ by merely computing successive bits of $p$, starting at the least significant bit.  This mechanism is so simple that it could be hidden, by accident or by malice, in a larger chemical reaction network.  Moreover, for suitable values of $m$ (e.g., $m = 34$, so that the threshold $p = 2^m$ is roughly $1.7 \times 10^{10}$),
 \begin{enumerate}
 	\item any attempt to model-check or simulate $\mathbf{N}_1$ will perforce use a population much less than the threshold and conclude that $\mathbf{N}_1$ will always turn blue; while
	\item any realistic wet-lab molecular implementation of $\mathbf{N}_1$ will have a population greater than the threshold and thus turn red.
\end{enumerate}
If the behaviors represented by blue and red here are a desired, ``good'' behavior of $\mathbf{N}_1$ (or of a network containing $\mathbf{N}_1$) and an undesired, ``bad'' behavior of this network, respectively, then the possibility of such a phase transition is a serious challenge to verifying the correct behavior of the chemical reaction network.  Simply put, this is a context in which the small population heuristic can lead us astray.

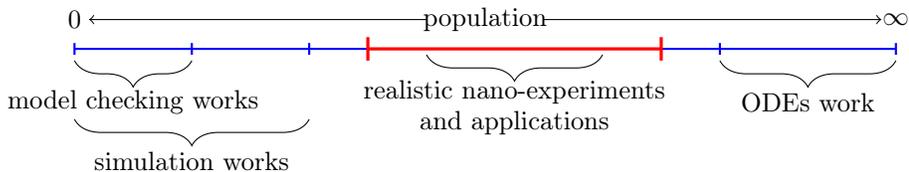
\begin{figure}[htb]
	\begin{tikzpicture}[scale=0.78]
\node at (7,0.5) {\footnotesize population};
\draw[<-] (0.25,0.5) -- (6,0.5);
\draw[->] (8,0.5) -- (13.75,0.5);
\node at (0,0.5) {\footnotesize $0$};
\node at (14,0.5) {\footnotesize $\infty$};

\draw[thick,blue] (0,0) -- (5,0);
\draw[very thick,red] (5,0) -- (10,0);
\draw[thick,blue] (10,0) -- (14,0);

\draw[thick,blue] (0,-0.1) -- (0,0.1);
\draw[thick,blue] (2,-0.1) -- (2,0.1);
\draw[thick,blue] (4,-0.1) -- (4,0.1);
\draw[very thick,red] (5,-0.2) -- (5,0.2);
\draw[very thick,red] (10,-0.2) -- (10,0.2);
\draw[thick,blue] (11,-0.1) -- (11,0.1);
\draw[thick,blue] (14,-0.1) -- (14,0.1);

\draw[decorate,decoration={brace,amplitude=10pt}] (9,-0.1) -- (6,-0.1) node [midway,yshift=-0.7cm,align=center] {\footnotesize realistic nano-experiments\\[-0.1cm]\footnotesize and applications};
\draw[decorate,decoration={brace,amplitude=10pt}] (2,-0.2) -- (0,-0.2) node [midway,yshift=-0.55cm] {\footnotesize model checking works};
\draw[decorate,decoration={brace,amplitude=10pt}] (4,-1.2) -- (0,-1.2) node [midway,yshift=-0.55cm] {\footnotesize simulation works};
\draw[decorate,decoration={brace,amplitude=10pt}] (14,-0.2) -- (11,-0.2) node [black,midway,yshift=-0.55cm] {\footnotesize ODEs work};
\end{tikzpicture}
	\caption{Scales at which different verification methods (simulation, model checking, and ODE's) work.  The gap in the middle shows the scale at which none of these methods will catch the ``produce blue'' behavior of the system design.  This gap is problematic because it is the scale of realistic programmed molecular systems. We show in Section~\ref{sec:TheoremProving} how such systems can be verified using automated theorem proving. \label{fig:introfigure}}
\end{figure}

There is a dual {\it large population heuristic} that is used even more often than the small population heuristic.  A theorem of Kurtz \cite{jKurt72, oAndKur11, oAndKur15} tells us that each {\it stochastic} chemical reaction network (the type of chemical reaction network used in our work here and in most of molecular programming) behaves, at sufficiently large populations, like the corresponding {\it deterministic} chemical reaction network.  Since the behavior of a deterministic chemical reaction network is exactly described by a system of ordinary differential equations, this means that we can use a mathematical software package to numerically solve this system and thereby understand the behavior of the original stochastic chemical reaction network at sufficiently large populations.

In Section~\ref{sec:CoupledPhaseTransitions} we add a single reaction to the chemical reaction network $\mathbf{N}_1$, creating a chemical reaction network $\mathbf{N}_2$ that we prove (in Theorem~\ref{thm:pleq2pownterm}) to exhibit two {\it coupled population-induced phase transitions} in the following sense.  If $p < 2^m$ or $p \geq 2^n$, then $\mathbf{N}_2$ terminates with essentially all its population blue, while if $2^m \leq p < 2^n$, then $\mathbf{N}_2$ terminates with essentially all its population red.  Thus $\mathbf{N}_2$ exhibits sharp phase transitions at the two population thresholds, $p = 2^m$ and $p = 2^n$.  These phase transitions are {\it coupled} in that exceeding the second threshold returns the behavior of $\mathbf{N}_2$ to its behavior below the first threshold.  For suitable values of $m$ and $n$ (e.g. $m = 34$ and $n = 67$ as above, so that the thresholds $p = 2^m$ and $p = 2^n$ are roughly $1.7 \times 10^{10}$ and $1.5 \times 10^{20}$), this implies (see Figure~\ref{fig:introfigure}) that
\begin{enumerate}
	\item any attempt to model-check or simulate $\mathbf{N}_2$ will perforce use a population much less than the smaller threshold and conclude that $\mathbf{N}_2$ will always turn blue, and
 	\item any use of ordinary differential equations to analyze the behavior of $\mathbf{N}_2$ will also conclude that $\mathbf{N}_2$ will always turn blue, a conclusion that is only valid for populations greater than the larger threshold, while
	\item any realistic wet-lab molecular implementation of $\mathbf{N}_2$ will have a population between the two thresholds and thus turn red.
\end{enumerate}
The chemical reaction network $\mathbf{N}_2$ thus exemplifies a class of contexts in which the small population heuristic and the large population heuristic can both lead us astray.

We emphasize that the phase transitions in the chemical reaction networks $\mathbf{N}_1$ and $\mathbf{N}_2$ occur at thresholds in their {\it absolute populations}.  In contrast, phase transitions in chemical reaction networks for approximate majority \cite{jAnAsEi08a, jCarCsi12, cCHKM17} occur at threshold {\it ratios between sub-populations}, and phase transitions in bacterial quorum sensing \cite{jMilBas01} occur at threshold {\it population densities}.

Section~\ref{sec:Implications} discusses the consequences of our results for the verification of programmed molecular systems in some detail.  Here we summarize these consequences briefly.  Phase transitions are ubiquitous in natural and engineered systems \cite{cRanKap10, cRand17, jCaMiRa18}.  Our results are thus cautionary, but they should not be daunting.  Fifteen years after Turing proved the undecidability of the halting problem, Rice \cite{oRice51, jRice53} proved his famous generalization stating that {\it every} nontrivial input/output property of programs is undecidable.  Rice's theorem saves valuable time, but it has never prevented computer scientists from developing specific programs in disciplined ways that enable them to be verified.  Similarly, Sections \ref{sec:SinglePhaseTransition} and \ref{sec:CoupledPhaseTransitions} give mathematical {\it proofs} that the chemical reaction networks $\mathbf{N}_1$ and $\mathbf{N}_2$ have the properties described above, and Section \ref{sec:Implications} describes how we have implemented such proofs in the Isabelle proof assistant \cite{oNiPaWe02, jPaNiWe19}.  As molecular programming develops, simulators, model checkers, theorem provers, and other tools will evolve with it, as will disciplined scientific judgment about how and when to use such tools.

\section{Chemical Reaction Networks \label{sec:CRN}}

Chemical reaction networks (CRNs) are abstract models of molecular processes in well-mixed solutions.  They are roughly equivalent to three models used in distributed computing, namely, Petri nets, population protocols, and vector addition systems \cite{oCSWB09}.  This paper uses stochastic chemical reaction networks.

For our purposes, a ({\it stochastic}) {\it chemical reaction network} {\bf N} consists of finitely many {\it reactions}, each of which has the form
\begin{equation}
\label{eqn:abcd}
A + B \rightarrow C + D,
\end{equation}
where $A$, $B$, $C$, and $D$ (not necessarily distinct) are {\it species}, i.e., abstract types of molecules.  Intuitively, if this reaction occurs in a solution at some time, then one $A$ and one $B$ disappear from the solution and are replaced by one $C$ and one $D$, these things happening instantaneously.  A {\it state} of the chemical reaction network {\bf N} with species $A_1,\dots,A_n$ at a particular moment of time is the vector $(a_1,\dots,a_s)$, where each $a_i$ is the nonnegative integer count of the molecules of species $A_i$ in solution at that moment.  Note that we are using the so called ``lower-case convention'' for denoting species counts.

In the full stochastic chemical reaction network model, each reaction also has a positive real {\it rate constant}, and the random behavior of {\bf N} obeys a continuous-time Markov chain derived from these rate constants.  However, our results here are so robust that they hold for {\it any} assignment of rate constants, so we need not concern ourselves with rate constants or continuous-time Markov chains.  In fact, for this paper, we can consider the reaction (\ref{eqn:abcd}) to be the if-statement
\begin{equation}
\label{eqn:rxndelta}
\text{if }a>0\text{ and }b>0\text{ then }a, b, c, d := a-1, b-1, c+1, d+1,
\end{equation}
where ``:='' is parallel assignment.   The reaction (\ref{eqn:abcd}) is \textit{enabled} in a state $q$ of $\mathbf{N}$ if $a>0$ and $b>0$ in $q$; otherwise, this reaction is \textit{disabled} in $q$.  A state $q$ of $\mathbf{N}$ is \textit{terminal} if no reaction is enabled in $q$.

A \textit{trajectory} of a chemical reaction network $\mathbf{N}$ is a sequence $\tau = (q_i\ |\ 0 \le i < \ell)$ of states of $\mathbf{N}$, where $\ell \in \mathbb{Z}^+ \cup \{\infty\}$ is the \textit{length} of $\tau$ and, for each $i \in \mathbb{N}$ with $i + 1 < \ell$, there is a reaction of $\mathbf{N}$ that is enabled in $q_i$ and whose effect, as defined by (\ref{eqn:rxndelta}), is to change the state of $\mathbf{N}$ from $q_i$ to $q_{i+1}$.  A trajectory $\tau = (q_i\ |\ 0 \le i < \ell)$ is \textit{terminal} if $\ell < \infty$ and $q_{\ell - 1}$ is a terminal state of $\mathbf{N}$.

Assume for this paragraph that the context specifies an initial state $q_0$ of $\mathbf{N}$, as it does in this paper.  A state $q$ of $\mathbf{N}$ is \textit{accessible} if there is a finite trajectory $\tau = (q_i\ |\ 0 \le i < \ell)$ of $\mathbf{N}$ with $q_{\ell-1} = q$.  A \textit{full trajectory} of $\mathbf{N}$ is a trajectory $\tau = (q_0\ |\ 0 \le i < \ell)$ that is either terminal or infinite.

The fact that each reaction (\ref{eqn:abcd}) has two \textit{reactants} ($A$ and $B$) and two \textit{products} ($C$ and $D$) means that $\mathbf{N}$ is a \textit{population protocol} \cite{jAAER07}.  This condition implies that the total population of all species never changes in the course of a trajectory.  If such a chemical reaction network has $s$ species and initial population $p$, its state space is thus the $(s-1)$-dimensional integer simplex
\begin{equation}
    \Delta^{s-1}(p) = \left\{ (a_1,\dots,a_s) \in \mathbb{N}^s\ \middle|\ \sum_{i=1}^s a_i = p\right\}.
\end{equation}
Note that $|\Delta^{s-1}(p)| = \binom{p+s-1}{s-1}$.  Of course, fewer than this many states may be reachable from a particular initial state of $\mathbf{N}$.

A full trajectory $\tau = (q_i\ |\ 0 \le i < \ell)$ of a CRN $\mathbf{N}$ is (\textit{strongly}) \textit{fair} \cite{jKwia89, oBaiKat08} if it has the property that, for every state $q$ and reaction $\rho$ that is enabled in $q$,
\begin{equation}
\label{eqn:fairness}
    (\exists^\infty i) q_i = q \implies (\exists^\infty j) [q_j = q\ \textrm{and $\rho$ occurs at $j$ in $\tau$}],
\end{equation}
where $(\exists^\infty i)$ means ``there exist infinitely many $i$ such that.''  Note that every terminal trajectory of $\mathbf{N}$ is vacuously fair, because it does not satisfy the hypothesis of (\ref{eqn:fairness}).

The stochastic kinetics of chemical reaction networks implies that, regardless of the rate constants of the reactions, for every population protocol $\mathbf{N}$ and every initial population $p$ of $\mathbf{N}$, there is a real number $\varepsilon>0$ such that, for every state $q$ of $\mathbf{N}$ and reaction $\rho$ that is enabled in $q$, the probability that $\rho$ occurs in $q$ depends only on $q$ and is at least $\varepsilon$.  This in turn implies that, with probability 1, $\mathbf{N}$ follows a fair trajectory.  Hence, if $\mathbf{N}$ has a given behavior on all fair trajectories, then $\mathbf{N}$ has that behavior with probability 1.

We use the following two facts in Section 4.  The first is an obvious consequence of the definition of fairness.
\begin{observation}
\label{obs:reactionfairness}
If $\tau = (q_i\ |\ 0 \le i < \ell)$ is a fair trajectory of a population protocol $\mathbf{N}$, then, for every reaction $\rho$ of $\mathbf{N}$,
\begin{equation}
(\exists^\infty i)[\textrm{$\rho$ is enabled in $q_i$}] \implies (\exists^\infty j)[\textrm{$\rho$ occurs at $j$ in $\tau$}].
\end{equation}
\end{observation}

A famous theorem of Harel \cite{jHare86,oKoze06} implies that the general problem of deciding whether a chemical reaction network terminates on all fair trajectories is undecidable.  Nevertheless, the following lemma gives a useful sufficient condition for termination on all fair trajectories.  This lemma undoubtedly follows from a very old result on fairness, but we do not know a proper reference at the time of this writing.  A proof appears in the Appendix.

\begin{lemma}[fair termination lemma]\label{lem:fairterm} If a population protocol with a specified initial state has a terminal trajectory from every accessible state, then all its fair trajectories are terminal.
\end{lemma}

\section{Single Phase Transition \label{sec:SinglePhaseTransition}}
\label{sec:singlethreshold}
This section presents the chemical reaction network $\mathbf{N}_1$ and proves that it exhibits a population-induced phase transition as described in the introduction.

Fix $m,n,p\in\mathbb{Z}^+$ with $n>m+1$. Let $\mathbf{N}_1$ be a chemical reaction network consisting of the $n+1$ \emph{$\zeta$-reactions}
\[\zeta_i \equiv Z_i+Z_i\to\begin{cases}Z_{i+1}+B&(0\leq i<m)\\Z_{i+1}+R&(m\leq i<n)\\Z_i+R&(i=n)\end{cases}\]
and the \emph{$\chi$-reaction}
\[\chi\equiv B+R\to R+R.\]
All results here hold regardless of the rate constants of these $n+2$ reactions.

We initialize $\mathbf{N}_1$ with $z_0=p$ and all other counts 0.

Intuitively, $B$ is \emph{blue}, $R$ is \emph{red}, and the species $Z_i$ are all \emph{colorless}.

\begin{lemma}\label{lem:1QS1}
	$\mathbf{N}_1$ terminates on all possible trajectories.
\end{lemma}

\begin{proof}
	Every reaction in $\mathbf{N}_1$ reduces the rank
	\[3\sum_{i=0}^n z_i+2b+r.\]
\end{proof}

\begin{notation}
	For $1\leq k\leq n+1$, let
	\[S_k=\sum_{i=0}^{k-1}2^iz_i,\]
	noting that this quantity depends on the state of $\mathbf{N}_1$.
\end{notation}

\begin{lemma}\label{lem:1QS2}
	Let $0\leq j\leq n$ and $1\leq k\leq n+1$.
	\begin{enumerate}
		\item If $j\neq k-1$, then the reaction $\zeta_j$ preserves the value of $S_k$.
		\item If $j=k-1$, then the reaction $\zeta_j$ reduces the value of $S_k$.
	\end{enumerate}
\end{lemma}

\begin{proof}
	Let $0\leq j\leq n$ and $1\leq k\leq n+1$.
	\begin{enumerate}
		\item\label{it:1QS1} This holds trivially if $k\leq j\leq n$, so assume that $0\leq j<k-1$. Let $z_0,\ldots,z_{k-1}$ be the counts of $Z_0,\ldots,Z_{k-1}$ just before an occurrence of $\zeta_j$, let $z'_0,\ldots,z'_{k-1}$ be the counts just after, and let $S'_k=\sum_{i=0}^{k-1}2^i z'_i$. Since $\zeta_j$ occurs, we must have $z_j\geq 2$. If we let $I=\{0,\ldots,k-1\}\setminus\{j,j+1\}$, then each
		\[z'_i=\begin{cases}z_i&\text{if }i\in I\\z_i-2&\text{if }i=j\\z_i+1&\text{if }i=j+1,\end{cases}\]
		so we have
		\[S'_k=\sum_{i\in I}2^i z_i+2^j(z_j-2)+2^{j+1}(z_{j+1}+1)=S_k.\]
		
		\item Assume that $j=k-1$, and use the notations $z_0,\ldots,z_{k-1}$, $z'_0,\ldots,z'_{k-1}$, and $S'_k$ as in part~\ref{it:1QS1} of this proof. If $1\leq k\leq n$, then each
		\[z'_i=\begin{cases}z_i&\text{if }0\leq i<k-1\\z_i-2&\text{if }i=k-1.\end{cases}\]
		If $k=n+1$ ,then each
		\[z'_i=\begin{cases} z_i&\text{if }0\leq i<k-1\\z_{i-1}&\text{if }i=k-1.\end{cases}\]
		Either way, $S'_k<S_k$.
	\end{enumerate}
\end{proof}

\begin{corollary}\label{cor:1QS3}
	For every $1\leq k\leq n+1$, the inequality $S_k\leq p$ is an invariant of $\mathbf{N}_1$.
\end{corollary}

\begin{proof}
	Let $1\leq k\leq n+1$. We have $S_k=p$ in the initial state that we have specified. The $\chi$-reaction trivially preserves the value of $S_k$ and Lemma~\ref{lem:1QS2} tells us that the reactions $\zeta_0,\ldots,\zeta_n$ all preserve the condition $S_k\leq p$.
\end{proof}

\begin{corollary}\label{cor:1QS4}
	If $1\leq k\leq n$ and $z_k>0$ in some accessible state of $\mathbf{N}_1$, then $p\geq 2^k$.
\end{corollary}

\begin{proof}
	Assume the hypothesis. Then Corollary~\ref{cor:1QS3} tells us that, in the given accessible state,
	\[p\geq S_{k+1}\geq 2^k z_k\geq 2^k.\]
\end{proof}

In the following, for $d\in\mathbb{Z}^+$, we use both the mod-$d$ \emph{congruence} (equivalence relation)
\[a\equiv b\bmod d,\]
which asserts of integers $a,b\in\mathbf{Z}$ that $b-a$ is divisible by $d$, and the \emph{$\mathbf{mod}$-$d$ operation}
\[b\;\mathbf{mod}\;d\]
whose value, for $b\in\mathbb{Z}$, is the unique $r\in\mathbb{Z}$ such that $0\leq r<d$ and $r\equiv b\bmod d$.

\begin{corollary}\label{cor:1QS5}
	The congruence
	\begin{equation}\label{eq:1QSA}
		S_n\equiv p\bmod 2^n
	\end{equation}
	is an invariant of $\mathbf{N}_1$.
\end{corollary}

\begin{proof}
	The initialization of $\mathbf{N}_1$ ensures that (\ref{eq:1QSA}) holds in the initial state. It is clear that the reactions $\zeta_n$ and $\chi$ preserve the value of $S_n$, and Lemma~\ref{lem:1QS2} tells us that the reactions $\zeta_1,\ldots,\zeta_{n-2}$ preserve the value of $S_n$. Hence it suffices to show that the reaction $\zeta_{n-1}$ preserves the truth of (\ref{eq:1QSA}).

	Assume that (\ref{eq:1QSA}) holds just prior to the occurrence of $\zeta_{n-1}$. Let $z_0,\ldots,z_{n-1}$ be the counts of $Z_0,\ldots,Z_{n-1}$ just before this occurrence, let $z'_0,\ldots,z'_{n-1}$ be the counts just after this occurrence, and let $S'_n=\sum_{i=0}^{n-1}2^i z'_i$. Since $\zeta_{n-1}$ occurs, we must have $z_{n-1}\geq 2$. For each $0\leq i<n$ we have
	\[z'_n=\begin{cases}z_i&\text{if }0\leq i<n-1\\z_i-2&\text{if }i=n-1,\end{cases}\]
	so
	\begin{align*}
		S'_n&=\sum_{i=0}^{n-2}2^i z_i +2^{n-1}(z_{n-1}-2)\\
		&=S_n-2^n\\
		&\equiv S_n\bmod 2^n.
	\end{align*}
	Since $S_n\equiv p\bmod 2^n$, it follows that $S'_n\equiv p\bmod 2^n$, i.e., that (\ref{eq:1QSA}) holds just after this occurrence of $\zeta_{n-1}$.
\end{proof}

\begin{corollary}\label{cor:1QS6}
	For every $1\leq k\leq n$, the condition
	\[\Theta_k\equiv z_k=\cdots=z_n=0\implies S_k=p\]
	is an invariant of $\mathbf{N}_1$.
\end{corollary}

\begin{proof}
	Let $1\leq k\leq n$. The condition $\Theta_k$ holds trivially in the initial state that we have specified. The reactions $\zeta_n$ and $\chi$ trivially preserve the value of $S_k$, so let $0\leq j<n$. It suffices to show that $\zeta_j$ preserves the condition $\Theta_k$. For this, assume that $\Theta_k$ holds just prior to an occurrence of $\zeta_j$. Let $z_0,\ldots,z_n$ be the counts of $Z_0,\ldots,Z_n$ just prior to this occurrence of $\zeta_j$, and let $z'_0,\ldots,z'_n$ be the counts just after this occurrence. Since $\zeta_j$ occurs, we must have $z_j\geq 2$ and $z'_{j+1}\geq 1$. To see that $\Theta_k$ holds just after this occurrence of $\zeta_j$, assume that $z'_k=\cdots=z'_n=0$. Then $z_k=\cdots=z_n=0$ and $j+1<k$, so Lemma~\ref{lem:1QS2} tells us that $\zeta_j$ preserves the value of $S_k$. Hence $\Theta_k$ holds just after this occurrence of $\zeta_j$.
\end{proof}

\begin{corollary}\label{cor:1QS7}
	Let $(q_0,\ldots,q_t)$ be a trajectory of $\mathbf{N}_1$, where $q_t$ is a terminal state, and let $1\leq k\leq n$. If $p\geq 2^k$, then there exists $1\leq s\leq t$ such that $z_k>0$ in $q_s$.
\end{corollary}

\begin{proof}
	Let $(q_0,\ldots,q_t)$ be a trajectory of $\mathbf{N}_1$, and let $1\leq k\leq n$. Assume that $p\geq 2^k$ and that there does not exist $s\in\{1,\ldots,t\}$ such that $z_k>0$. It suffices to show that $q_t$ is not terminal.

	Since there does not exist $s\in\{1,\ldots,t\}$ such that $z_k>0$ in $q_s$, it must be the case that $z_k=\cdots=z_n=0$ in $q_t$. It follows by Corollary~\ref{cor:1QS6} that $S_k=p$ holds in $q_t$. Since $\sum_{i=0}^{k-1}z_i=2^{k}-1$ and we have $p\geq 2^k$ by assumption, this implies that there exists $0\leq i<k$ such that $z_i>1$ in $q_t$. Hence the reaction $\zeta_i$ is enabled in $q_t$, so $q_t$ is not terminal.
\end{proof}

\begin{notation}
	For each $r\in\{0,\ldots,2^n-1\}$, let $\lambda(r)$ be the number of 1s in the $n$-bit binary representation of $r$ (leading 0s allowed), and let
	\[\varepsilon=\begin{cases}\lambda(p)&\text{if }p<2^n\\1+\lambda(p\;\mathbf{mod}\;2^n)&\text{if }p\geq 2^n.\end{cases}\]
	Note that $\varepsilon$ is an integer depending on $n$ and $p$, and that $\varepsilon$ is negligible in the sense that $\varepsilon=o(p)$ as $p\to\infty$.
\end{notation}
The boolean value of a condition $\varphi$ is $\llbracket\varphi\rrbracket=\textbf{if }\varphi\textbf{ then }1\textbf{ else }0$.

\begin{theorem}\label{thm:1QS8}
	$\mathbf{N}_1$ terminates on all trajectories in the state $(z_0,\ldots,z_n,b,r)$ specified as follows.
	\begin{enumerate}
		\item\label{it:1QSi} $z_{n-1}\cdots z_0$ is the $n$-bit binary expansion of $p\;\mathbf{mod}\;2^n$.
		\item\label{it:1QSii} $z_n=\llbracket p\geq 2^n\rrbracket$.
		\item\label{it:1QSiii} $b=(p-\varepsilon)\cdot\llbracket p< 2^m\rrbracket$
		\item\label{it:1QSiv} $r=(p-\varepsilon)\cdot\llbracket p\geq 2^m\rrbracket$.
	\end{enumerate}
\end{theorem}

\begin{proof}
	Lemma~\ref{lem:1QS1} tells us that $\mathbf{N}_1$ terminates on all trajectories. Let $q=(z_0,\ldots,z_n,b,r)$ be a terminal state of $\mathbf{N}_1$, and note the following.
	\begin{enumerate}
		\item\label{it:1QSa} For all $0\leq i\leq n$, $\zeta_i$ is not enabled in $q$, so $z_i\in\{0,1\}$.
		\item\label{it:1QSb} $\chi$ is not enabled in $q$, so $b=0$ or $r=0$.
		\item\label{it:1QSc} By (\ref{it:1QSa}), $S_n\leq \sum_{i=0}^{n-1}2^i=2^n-1$, so Corollary~\ref{cor:1QS5} tells us that $S_n=p\;\mathbf{mod}\;2^n$, i.e., that (\ref{it:1QSi}) holds.
		\item\label{it:1QSd} If $p<2^n$, then Corollary~\ref{cor:1QS4} tells us that $z_n=0$. If $p\geq 2^n$, then Corollary~\ref{cor:1QS7} tells us that $z_n\geq 1$ somewhere along every trajectory leading to $q$. Since $z_n$ can never become 0 after becoming positive, this implies that $z_n=1$ in $q$. Hence (\ref{it:1QSii}) holds.
		\item\label{it:1QSe} By (\ref{it:1QSc}) and (\ref{it:1QSd}) we have $\sum_{i=0}^n z_i=\varepsilon$.
		\item\label{it:1QSf} Since $b+r+\sum_{i=0}^n z_i$ is an invariant of $\mathbf{N}_1$, (\ref{it:1QSb}) and (\ref{it:1QSe}) tell us that one of $b$ and $r$ is $p-\varepsilon$ and the other is 0.
		\item \label{it:1QSg} If $p<2^m$, then Corollary~\ref{cor:1QS4} tells us that $z_m=\cdots=z_n=0$ holds throughout every trajectory leading to $q$. This implies that none of the reactions $\zeta_m,\ldots,\zeta_n$ occurs along any trajectory leading to $q$, whence $r=0$.
		\item \label{it:1QSh} If $p\geq 2^m$, then Corollary~\ref{cor:1QS7} tells us that $z_m>0$ holds somewhere along every trajectory leading to $q$. This implies that the reaction $\zeta_{m-1}$ occurs, whence $r$ becomes positive, somewhere along every trajectory leading to $q$. Since $r$ can never become 0 after becoming positive, this implies that $r>0$.
		\item By (\ref{it:1QSf}), (\ref{it:1QSg}), and (\ref{it:1QSh}), (\ref{it:1QSiii}) and (\ref{it:1QSiv}) hold.
	\end{enumerate}
\end{proof}

Since $\varepsilon$ is negligible with respect to $p$, Theorem~\ref{thm:1QS8} says that $\mathbf{N}_1$ terminates in an overwhelmingly blue state if $p < 2^m$ and in an overwhelmingly red state if $p \geq 2^m$.  This is a very sharp phase transition at the population threshold $2^m$.

\section{Coupled Phase Transitions \label{sec:CoupledPhaseTransitions}}
Let $m, n, p$, and $\mathbf{N_1}$ be as in Section~\ref{sec:singlethreshold}, and let $\mathbf{N_2}$ be a CRN consisting of the $n+2$ reactions of $\mathbf{N_1}$ and the \textit{$\omega$-reaction}
\begin{equation*}
    \omega \equiv R + Z_n \rightarrow B + Z_n.
\end{equation*}
This section proves that $\mathbf{N_2}$ exhibits two coupled population-induced phase transitions as described in the introduction.

We use the same initialization for $\mathbf{N_2}$ as for $\mathbf{N_1}$.  Again, all our results hold regardless of the rate constants of the $n+3$ reactions of $\mathbf{N_2}$.

Routine inspection verifies the following.
\begin{observation}
Lemma~\ref{lem:1QS2} and Corollaries~\ref{cor:1QS3}-\ref{cor:1QS7} hold for $\mathbf{N_2}$ as well as for $\mathbf{N_1}$.
\end{observation}

If $p<2^n$, then Corollary~\ref{cor:1QS4} tells us that $z_n$ never becomes positive in $\mathbf{N_2}$, so the $\omega$-reaction never occurs in $\mathbf{N_2}$.  Thus, for $p<2^n$, $\mathbf{N_2}$ behaves exactly like $\mathbf{N_1}$.

On the other hand, if $p \ge 2^n$, then the behavior of $\mathbf{N_2}$ is very different from that of $\mathbf{N_1}$.  For example, in contrast with Lemma \ref{lem:1QS1}, we have the following.
\begin{lemma}
\label{lem:lemma12}
If $p \ge 2^n$, then not all trajectories of $\mathbf{N_2}$ terminate.
\end{lemma}

\begin{proof}
    Assume that $p \geq 2^n$.
    Let
    \[\vec{\zeta} = \zeta_0^{2^{n-1}}\zeta_1^{2^{n-2}}\dots\zeta_{n-1}^{2^0}\]
    denote a sequence consisting of $2^{n-1}$ consecutive occurrences of the reaction $\zeta_0$, followed by $2^{n-2}$ occurrences of $\zeta_1$, etc.
    Since $p \geq 2^n$ each of these $2^n-1$ reactions is enabled when its turn comes.
    After the sequence $\vec{\zeta}$ has occurred, we have
    \begin{align*}
        z_0 &= p - 2^n,\\
        z_1 &= \dots = z_{n-1} = 0,\\
        z_n &= 1,\\
        b &= \sum_{i = 0}^{m-1} 2^{n-(i+1)} = 2^n (1-2^{-m}),\\
        r &= \sum_{i = m}^{n-1} 2^{n-(i+1)} = 2^{n-m} - 1. 
    \end{align*}
    Recalling that $n > m + 1$, we have $r \geq 3$ here, so the reaction $\omega$ is enabled after $\vec{\zeta}$ has occurred.
    In fact $\vec{\zeta}$ can (in principle) be followed by the infinite sequence
    \[\omega, \chi, \omega, \chi, \dots\]
    of reactions, so $\bfNt$ has a nonterminating trajectory.
\end{proof}

It is easy to see that the infinite trajectory of $\bfNt$ exhibited in the proof of Lemma \ref{lem:lemma12} is not fair.
In fact, we prove below that all fair paths of $\bfNt$ terminate.
First, however, we note that $\bfNt$, like $\bfN$, has a unique terminal state.

Let $\varepsilon$ be as defined before Theorem~\ref{thm:1QS8}.
\begin{lemma}
    \label{lem:lemma13}
    If $p \geq 2^n$ and $\bfNt$ terminates, then it does so in the state $(\zs,b,r)$ specified as follows.
    \begin{enumerate}[(i)]
        \item $z_{n-1}\cdots z_0$ is the $n-$bit binary expansion of $p \bfmod 2^n$.
        \item $z_n = 1$.
        \item $b = p - \varepsilon$.
        \item $r = 0$.
    \end{enumerate}
\end{lemma}
\begin{proof}
    Let $q$ be an accessible state of $\bfNt$ that is terminal.
    The proofs of (i) and (ii) are the same as in Theorem~\ref{thm:1QS8}, together with the fact that the $\omega-$reaction does not alter the value of $z_n$.
    Since $z_n = 1$ and the $\omega-$reaction is disabled in state $q$, (iv) holds in $q$.
    Finally, since
    \[p = \sum_{i=0}^n z_i + b + r = \varepsilon + b + r,\]
    (iii) follows from (i), (ii), and (iv).
\end{proof}
\begin{lemma}
    \label{lem:fairimpliesdisable}
    On any fair trajectory of $\bfNt$, after finitely many steps, all $\zeta$-reactions are permanently disabled.
\end{lemma}
\begin{proof}
    For each $0 \leq j \leq n$, let $\Phi_j$ be the assertion that, on any fair trajectory of $\bfNt$, after finitely many steps the reaction $\zeta_j$ is permanently disabled.
    It suffices to prove that $\Phi_j$ holds for all $0 \leq j \leq n$.
    We do this by induction on $j$.

    Let $0 \leq j \leq n$ and assume that $\Phi_i$ holds for all $0 \leq i < j$.
    It suffices to show that $\Phi_j$ holds.
    For this, let $\tau = (q_i | 0 \leq i \leq \infty)$ be a trajectory in which $\zeta_j$ is enabled infinitely often.
    It suffices to show that $\tau$ is not fair.
    We have two cases.
    \begin{case}
        Some $\zeta_i$ for $0 \leq i < j$ is enabled infinitely often.
        Then $\tau$ is not fair by the induction hypothesis.
    \end{case}
    \begin{case}
        There exists $k^* \in \mathbb{N}$ such that, for all $k \geq k^*$ the reactions $\zeta_i$, for $0 \leq i < j$, are all disabled in $q_k$.
        Then $z_j$ does not increase at any step of $\tau$ from $k^*$ onward.
        Since every occurrence of $\zeta_j$ decreases $z_j$, this implies that $\zeta_j$ only occurs finitely many times after $k^*$, hence only finitely many times in $\tau$.
        Since $\zeta_j$ is enabled infinitely often along $\tau$ it follows by Observation 1 that $\tau$ is not fair. \qedhere
    \end{case}
    \end{proof}
\begin{lemma}
    \label{lem:chiomegaterminal}
    With any initialization, all fair trajectories of the chemical reaction network $\textbf{N}{\chi\omega}$, consisting of just the reactions $\chi$ and $\omega$, are terminal.
\end{lemma}
\begin{proof}
    If $\textbf{N}{\chi\omega}$ is initialized with $z_n = 0$, then there is only one full trajectory, which is terminal, so it suffices to prove the lemma for initializations with $z_n > 0$.
    Let $q_0$ be any such initial state, and let $\hat{p}$ be the value of $b+r+z_n$ in $q_0$.
    Since $z_n$ and $\hat{p}$ are invariants of $\textbf{N}{\chi\omega}$, a state of $\textbf{N}{\chi\omega}$ is completely determined by the value of $r$.
    We thus refer to ``the state $r$'' of $\textbf{N}{\chi\omega}$, for $0 \leq r \leq \hat{p}-z_n$.
    Note that in this terminology the unique terminal state is 0.

    For each state $r$ of $\textbf{N}{\chi\omega}$ with initial state $q_0$ the trajectory $\tau_r = (r, r-1, \dots, 1, 0)$ given by $r$ consecutive occurrences of $\omega$ is a terminal trajectory from $r$, so the fair termination lemma (Lemma~\ref{lem:fairterm}) tells us that all fair trajectories of $\textbf{N}{\chi\omega}$ are terminal.
\end{proof}

Recall the notation defined just before Theorem~\ref{thm:1QS8}.  The following result is our main theorem.

\begin{theorem}
    \label{thm:pleq2pownterm}
    Let $(z_0,\dots,z_n,b,r)$ be the state of $\bfNt$ specified as follows.
    \begin{enumerate}[(i)]
        \item $z_{n-1}\cdots z_0$ is the $n$-bit binary expansion of $p \bfmod 2^n$.
        \item $z_n = \llbracket p \geq 2^n \rrbracket$.
        \item $b = (p - \varepsilon) * \llbracket p < 2^m \text{ or } p \geq 2^n \rrbracket$.
        \item $r = (p - \varepsilon) * \llbracket 2^m \leq p < 2^n \rrbracket$.
    \end{enumerate}
    If $p < 2^n$, then $\bfNt$ terminates in this state on all trajectories.
    If $p \geq 2^n$, then $\bfNt$ terminates in this state on all fair trajectories.
\end{theorem}
\begin{proof}
    If $p < 2^n$, then Corollary~\ref{cor:1QS3} tells us that $z_n$ never becomes positive in $\bfNt$, so $\omega$ is never enabled.
    Hence, in this case $\bfNt$ behaves exactly like $\bfN$.
    Theorem~\ref{thm:1QS8} tells us that $\bfNt$ terminates on all trajectories to the state satisfying (i) and (ii) above and, since $\llbracket p < 2^m \rrbracket = \llbracket p < 2^m \text{ or } p \geq 2^n \rrbracket$ and $\llbracket p \geq 2^m \rrbracket = \llbracket 2^m \leq p < 2^n \rrbracket$, also satisfying (iii) and (iv) above.

    If $p \geq 2^n$, then Lemmas~\ref{lem:fairimpliesdisable} and \ref{lem:chiomegaterminal} together tell us that $\bfNt$ terminates on all fair trajectories.
    Since $\llbracket p \geq 2^n \rrbracket = 1,  \llbracket p < 2^m \text{ or } p \geq 2^n \rrbracket =1$, and $\llbracket 2^m \leq p < 2^n \rrbracket = 0$, Lemma~\ref{lem:lemma13} tells us that termination must occur in the state satisfying (i)-(iv) above.
\end{proof}

Since $\varepsilon$ is again negligible with respect to $p$, Theorem~\ref{thm:pleq2pownterm} says that $\bfNt$ terminates in an overwhelmingly blue state if $p < 2^m$ or $p \geq 2^n$ but in an overwhelmingly red state if $2^m \leq p < 2^n$.  Hence $\bfNt$ exhibits very sharp phase transitions at the population thresholds $2^m$ and $2^n$.  As noted in the Introduction and elaborated in Section~\ref{sec:Implications} below, this has significant implications for the verification of chemical reaction networks.
\section{Implications for Verification \label{sec:Implications}}

The coupled phase transitions in the chemical reaction network $\mathbf{N_2}$ make it difficult to verify its behavior.  In this section we describe the use and limitations of verifying the chemical reaction network using simulation, model checking and differential equations.   None of these methods detected that the system turned red when the population reached $2^m$.  We then describe how the use of an interactive theorem prover enabled us to verify the chemical reaction network's behavior at both phase transitions, i.e., that it turned from blue to red at $2^m$ and from red to blue at $2^n$.   The fact that theorem proving could verify behavior that was otherwise not verified for the  chemical reaction network suggests that interactive theorem proving may have a useful role to play in future verification of a class of chemical reaction networks. 
Recall that the chemical reaction networks $\mathbf{N_1}$ and $\mathbf{N_2}$ have fixed populations throughout any given execution, and that their initial states have $z_0$ as the entire population. 

\subsection {Simulation}

The MATLAB SimBiology package is widely used to explore the behavior of a number of devices (molecules) executing concurrently \cite{matlab2019}. Using SimBiology, simulations of the $\mathbf{N_2}$ chemical reaction network were performed on an Intel processor computer with a processor clock of 5.0 GHz and 64GB of RAM.
Several simulations were performed with increasing  populations $z_0$.
With a population of $10^{7}$, the simulation performed as expected.
However, with a population of $10^{8}$, the simulation failed and terminated with no output or error message.
Thus, 
the stochastic simulation was unable to detect that the behavior of the $\mathbf{N_2}$ chemical reaction network could experience a phase transition.

\subsection{Model Checking}

The chemical reaction network $\mathbf{N_2}$  simulated in SimBiology and described above also was verified using the PRISM 4.6 probabilistic model checker \cite{cKwNoPa11}.  
Kwiatkowska and Thachuk, among others, have described the use of PRISM for the probabilistic verification of chemical reaction networks for biological systems \cite{jKwiTha14}.

To verify the chemical reaction network behavior we first converted the $\mathbf{N_2}$   model to SBML using the export function in SimBiology, and then converted the SBML model to PRISM using the sbml2prism conversion tool supplied with the PRISM software.  PRISM was used to verify six key properties of the $\mathbf{N_2}$ chemical reaction network at multiple populations. For example, one of the properties stated that  ``$P>=1[\text{F G } r = 0]$'', i.e., that with probability 1, the eventual state of the $R$ species has $0$ molecules, and never changes from that.
With a population of 100, PRISM generated the CTMC state model in 1.65 seconds using the same processor and memory as for the SimBiology simulations, and
the verification of the six properties required less than 2 seconds of CPU time.

For a population of 100 molecules, 97 correctly turn blue and 0 correctly turn red, since the latter only happens when the initial population is larger than $2^{34}$. PRISM also verified that in the final state the species count was $z_0 = 0$, $z_1 = 0$, $z_2 = 1$ and $z_5 = 1$, based on the binary expansion of one hundred.

However, we were unable to model check $\mathbf{N_2}$ with a population of 400
due to the rapid increase in states and limited memory. Thus, model checking confirmed the expected behavior of the $\mathbf{N_2}$ chemical reaction network for a population of $100$ but could not detect the behavioral change to red when the population increased. 

Advanced methods to prune a model so that meaningful model checking can occur include symmetry reduction \cite{cHKNPT06},  statistical model checking \cite{jCaKwWh18}, and automated partial exploration of the model \cite{jPaBrUc16}.  Recent work by Cauchi, et al. using formal synthesis allowed verification of systems with 10 continuous variables \cite{cCLLAKC19}.
However, even these methods would not be likely to help with the exceedingly large number of states  when the number of molecules is scaled to a realistic value for experiments. 
\color{black}

\subsection{Differential Equations}

We have seen how model checking and simulation fail to detect the ``red'' behavior in our chemical reaction network $\mathbf{N_2}$ due to the processing time and memory required for a large population.  The red behavior also is not detected when $\mathbf{N_2}$  
is approximated by deterministic semantics, and in this case it is not due to computational power.  In this model, a chemical reaction network is represented by a system of polynomial autonomous differential equations.   A mathematical connection between the stochastic model and the deterministic model is given by Kurtz \cite{jKurt72}, where he shows that the stochastic model is equivalent to the deterministic model as the number of molecules in the system grows towards infinity.   

In general, the system of differential equations induced by a chemical reaction network is difficult or impossible to solve exactly, and numerical methods are often used to approximate solutions. Here, we utilized MATLAB and the SimBiology package \cite{matlab2019} to numerically integrate the system of differential equations for $\mathbf{N_2}$.  The result found that $\mathbf{N_2}$  terminated in a predominantly blue state, again missing the red behavior. 

\subsection{Theorem Proving}
\label{sec:TheoremProving}

The simulation, model checking, and differential equations approaches to chemical reaction network verification outlined above all make some simplifying assumptions: reduced state space or generalization to the continuum.
In the case of our chemical reaction network, these assumptions lead to an incorrect verification result.

Interactive theorem proving, however, offers an exact approach that is guaranteed to apply at every scale.
In the interactive theorem proving paradigm, users create a machine-checkable mathematical proof of verification properties in collaboration with a software system.
Model checking also constructs a mathematical proof of correctness, but it relies more on a complete or semi-complete search of the state space in question.
By contrast, the goal of interactive theorem proving is to construct a more traditional mathematical proof that is also machine-checkable.
The result then applies to any population scale; a mathematical proof parameterized by population $N$ is valid at every possible value of $N$.

In a typical interactive theorem proving session, a user starts with a base of trusted facts generated from axioms and assumptions, and uses well-understood rules like modus ponens and double negation removal to construct new trusted facts and lemmas.
As with a conventional mathematical proof, the user's goal is to add new trusted facts in a strategic way until reaching the goal of the proof.

We have verified our chemical reaction network with Isabelle/HOL \cite{oNipKle14,oNiPaWe02}, a popular interactive theorem prover with several useful proof automation features.
Instead of working at the level of rules like modus ponens, users can instruct Isabelle to execute more general proof methods that can apply sequences of basic rules without user direct input.
For example, Isabelle can often prove the equivalence of predicate logic formulas with only one user-generated method invocation.
Once invoked, such a method attempts to automatically construct a series of low-level logical rules whose application proves the equivalence.
An Isabelle proof, then, consists of a directed acyclic graph of facts, connected by applications of these methods.
The user's task is to choose a chain of intermediate goal facts in a way that allows Isabelle to connect them easily on the way to the overall goal.

Isabelle also provides the powerful Sledgehammer automation tool, which makes calls to external proof systems to automate aspects of proof creation.  Sledgehammer takes a goal fact as input and attempts to generate a method invocation that proves it, operating at one level of abstraction above the proof methods invocations discussed above.
Since it is often unclear which method to invoke (or which arguments to supply to it), this functionality can increase proof construction speed substantially.

\begin{figure}
	\begin{center}
		\framebox{\includegraphics[width=3.0in,bb=0 0 645 615]{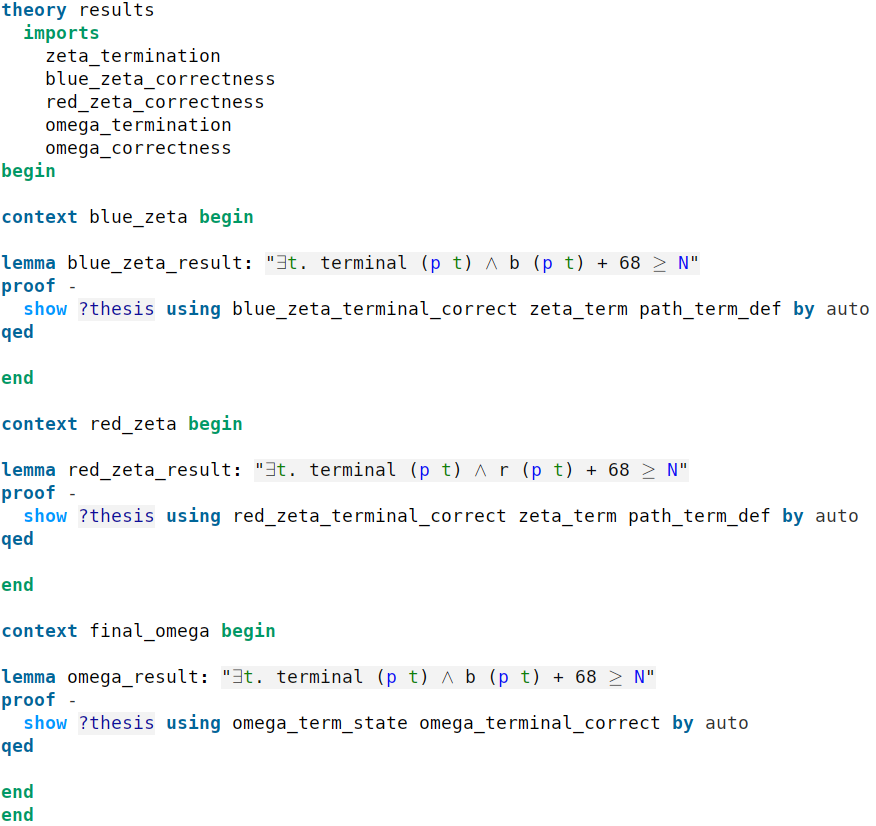}}
		\caption{\label{fig:isabelleresult}
		The end of the Isabelle proof, which summarizes its results in three lemmas.  The \texttt{context} statements bring our assumptions about the value of $N$ into context.  The \texttt{using} statements bring in trusted facts from the rest of our proof and supply them as arguments to Isabelle's \texttt{auto} proof method.  The identifier \texttt{p} refers to an arbitrary trajectory that is part of each context.  Isabelle displays all statements with a white or light gray background to indicate that it has checked them completely, and they are valid.}
	\end{center}
\end{figure}

We have used Isabelle to verify that our chemical reaction network has the desired behavior for all possible initializations.
That is, if we initialize it with $N<2^{34}$ or $N \ge 2^{67}$, the chemical reaction network terminates with majority blue, but if we initialize it with $2^{34} \le N < 2^{67}$, it terminates with majority red.
As expected, theorem proving is able to verify behavior correctly in all regions, including the middle region that is inaccessible to model checking, simulation, and ODE methods.
Figure~\ref{fig:isabelleresult} shows an image taken from the end of our Isabelle proof; it contains the three goal facts that we successfully verified, which summarize the behavior of the chemical reaction network.

Our Isabelle proof is loosely based on the proofs presented in Sections~\ref{sec:SinglePhaseTransition} and \ref{sec:CoupledPhaseTransitions}.
Whereas those proofs define two chemical reaction networks $\mathbf{N_1}$ and $\mathbf{N_2}$, we use Isabelle's \textit{locale} feature to associate assumptions about the population of $N$ with various parts of our proof.
In the locale where $N < 2^{35} $, for example, we are able to prove that our chemical reaction network terminates with majority blue.
Figure~\ref{fig:isabelleresult} shows how we enter these locales at the end of the proof to bring together our final results.

We refer to the three final locales as the lower blue region, the middle red region, and the upper blue region.
For each region, our proof must show both termination and correctness; i.e., we must show that our chemical reaction network reaches a final state where no reactions are possible, and that any possible final state has the specified red or blue population.
\begin{figure}
	\begin{center}
		\framebox{\includegraphics[height=3.0in]{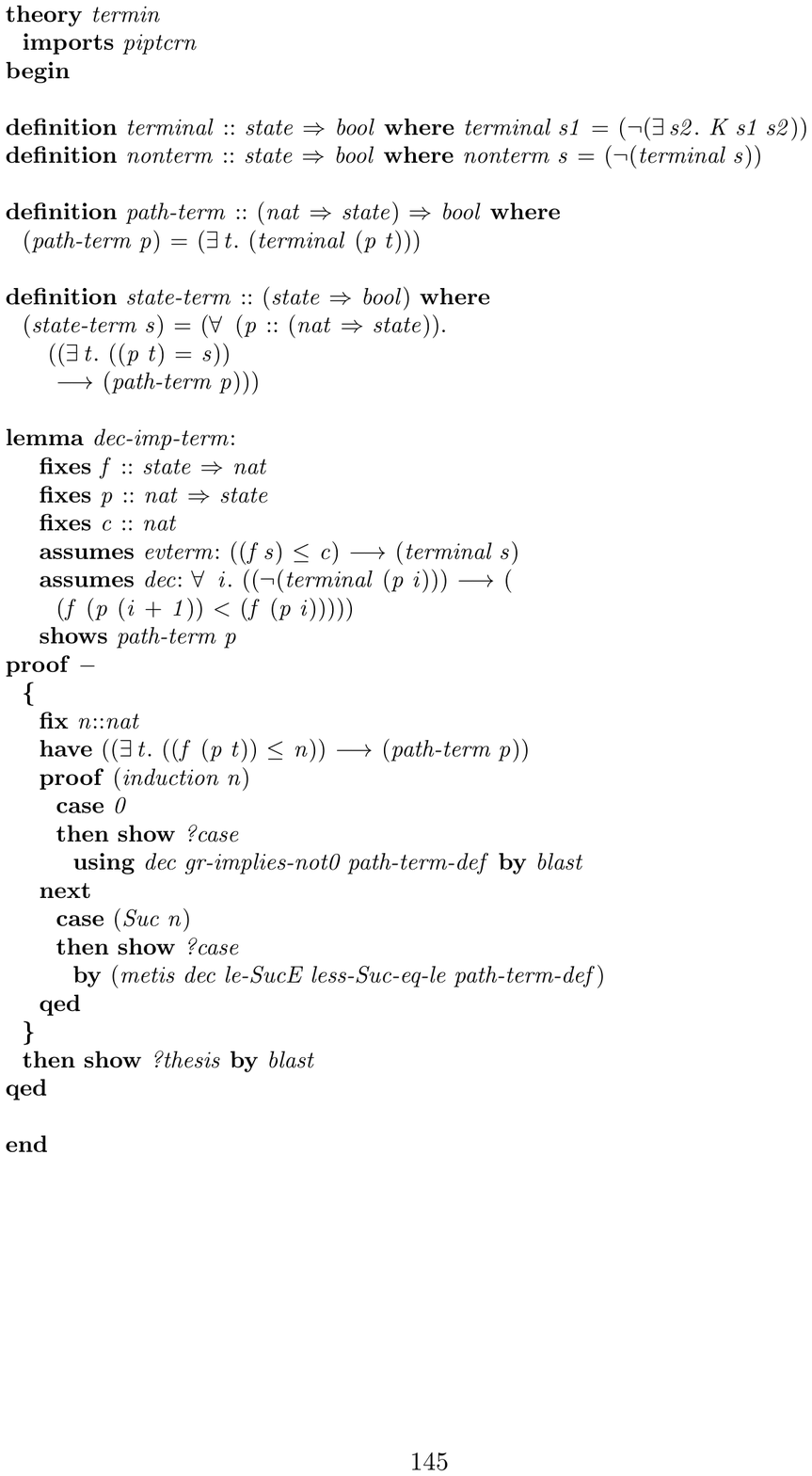}}
		\caption{\label{fig:isabellecode} This Isabelle code defines a terminal state as a state with no outgoing reactions; $K$ is a relation that encodes which state transitions our reaction set allows.  We also show a sample lemma that helps prove termination: if we identify a countdown expression $f$ and a constant $C$ such that all states with $f < C$ are terminal, then our system is guaranteed to terminate.}
	\end{center}
\end{figure}

As in Lemma~\ref{lem:1QS1}, we show termination in the lower two regions via a ``countdown'' expression that is guaranteed to decrease with every reaction.
See Figure~\ref{fig:isabellecode} for our Isabelle definitions of termination and a general lemma we proved that allows us to use the countdown technique.
In the upper blue region, it is impossible to prove termination without assuming that executions are fair.
Our Isabelle proof includes Equation~\ref{eqn:fairness} as an unproven assumption; we are not interested in unfair trajectories, but since they exist we cannot prove that all trajectories are fair.
For convenience, we also include Observation~\ref{obs:reactionfairness} as an an assumption. 
These two fairness assumptions allow us to prove that our chemical reaction network terminates in the upper blue region as well.

Our correctness proofs rely heavily on the sum
$S_{68}=\sum_{i=0}^{67} 2^i z_i$, using the notation of Section~\ref{sec:SinglePhaseTransition}, which is an invariant in the lower two regions.
In the upper blue region, it is an invariant until at least one $Z_{67}$ is produced.
This invariant allows us to reason about the composition of terminal states.
In the lower blue region, for example, we know that no red can ever be produced; the chemical reaction network can only produce its first red molecule alongside $Z$ species that would make the invariant too large.
Following the proof of Theorem~\ref{thm:1QS8}, then, we prove that any terminal state must be majority blue.

\section{Conclusion}
Taken together, the near-ubiquity of phase transitions in nature \cite{cRanKap10, jCaMiRa18}, the sheer size of molecular populations, and the simplicity of the chemical reaction networks that we have shown to exhibit population-induced phase transitions, indicate that molecular programming will present us with many exceptions to the otherwise useful notion that most bugs can be demonstrated with small counterexamples.  As we have seen, this presents a significant challenge to the verification of chemical reaction networks.  Here we suggest some directions of current and future research that might help meet this challenge.

A great deal of creative work has produced a steady scaling up of model checking to larger and larger state spaces \cite{jClEmSi09, jCDKB18, oAbSiTa18, jBCKL19, cLomPir19, cCJJKVZ19}.  Perhaps the most hopeful approach for dealing with population-induced phase changes, or with more general population-sensitive behaviors, is the model checking of parametrized systems \cite{oAbSiTa18}. 

Our results clearly demonstrate the advantage of including theorem proving (by humans and by software) in the verification toolbox for chemical reaction networks and other molecular programming languages.  This in turn suggests that software proof assistants such as Isabelle \cite{oNiPaWe02,oNipKle14} be augmented with features to deal more directly with chemical reaction networks and with population-sensitive phenomena.  It would also be useful to know how much of such work could be carried out with automated theorem provers such as Vampire \cite{cKovVor13}. 

Some future programmed molecular applications will be safety-critical, such as in health diagnostics and therapeutics. It is likely that evidence that such systems behave as intended will be required for certification by regulators prior to deployment.  Toward providing such evidence, Nemouchi et al. have recently shown how a descriptive language for safety cases can be incorporated into Isabelle in order to formalize argument-based safety assurance cases \cite{cNFGK19}.

We conclude with a more focused, theoretical question.  Our chemical reaction network $\mathbf{N}_1$ exhibits its phase transition on all trajectories, while $\mathbf{N}_2$ exhibits its coupled phase transitions only on all fair trajectories.  Is there a chemical reaction network that achieves $\mathbf{N}_2$'s coupled phase transitions on {\it all} trajectories?

\section*{Acknowledgments} We thank Neil Lutz for technical assistance.  The second and third authors thank Erik Winfree for his hospitality while they did part of this work during a 2020 sabbatical visit at the California Institute of Technology.  

\bibliographystyle{plain}
\bibliography{master}

\newpage
\appendix
\section{Proof of Fair Termination Lemma}

\begin{lemma}[fair termination lemma] 
If a population protocol with a specified initial state has a terminal trajectory from every accessible state, then all its fair trajectories are terminal.
\end{lemma}
\begin{proof}
Let $\mathbf{N}$ be a population protocol with initial state $q_0$, and assume that $\mathbf{N}$ has a terminal trajectory from every accessible state.  Let $\tau = (q_i\ |\ 0 \le i < \infty)$ be an infinite trajectory of $\mathbf{N}$.  It suffices to show that $\tau$ is not fair.

For each state $q$ of $\mathbf{N}$, let
\begin{equation}
    I_q = \{i \in \mathbb{N}\ |\ q_i = q\}.
\end{equation}
Since $\mathbf{N}$ is a population protocol, it has finitely many accessible states, so there is a state $q^*$ of $\mathbf{N}$ such that the set $I_{q^*}$ is infinite.  This state $q^*$ is accessible, so our assumption tells us that there is a finite trajectory $\tau^* = (q_i^*\ |\ 0 \le i < \ell)$ of $\mathbf{N}$ such that $q_0^* = q^*$ and $q^*_{\ell-1}$ is terminal.

Now $I_{q_0^*} = I_{q^*}$ is infinite and $I_{q_{\ell-1}^*} = \emptyset$ (because $q_{\ell-1}^*$ is terminal, so it does not appear in the infinite trajectory $\tau$), so there exists $0 \le k < \ell - 1$ such that $I_{q_k^*}$ is infinite and $I_{q_{k+1}^*}$ is finite.  Let $q^{**} = q_k^*$, and let $\rho$ be the reaction that takes $q_k^*$ to $q_{k+1}^*$.  Then $\rho$ is enabled in $q^{**}$ and there exist infinitely many $i$ such that $q_i = q^{**}$ (because $I_{q^{**}}$ is infinite), but there are only finitely many $j$ for which $q_j = q^*$ and $\rho$ occurs at $j$ in $\tau$ (because $I_{q_{k+1}^*}$ is finite).  Hence $\tau$ is not fair.
\end{proof}

\end{document}